\documentclass[runningheads,a4paper]{llncs}

\usepackage{amssymb}
\usepackage{graphicx}
\usepackage[margin=1.4in]{geometry}
\usepackage{url}
\newcommand{\keywords}[1]{\par\addvspace\baselineskip\noindent\keywordname\enspace\ignorespaces#1}
\usepackage{comment}
\usepackage[toc,page]{appendix}

\newtheorem{Claim}{Claim}

\usepackage{algorithm}
\usepackage{algorithmicx}
\usepackage{algpseudocode}

\begin{document}

\mainmatter  % start of an individual contribution


\title{Approximating the $k$-Set Packing Problem by Local Improvements}

%\titlerunning{Lecture Notes in Computer Science: Authors' Instructions}

\author{Martin F\"{u}rer
\and Huiwen Yu}


\institute{
 Department of Computer Science and Engineering\\
 The Pennsylvania State University, University Park, PA, USA.\\
 }
%\emailpsu \\
%\url{http://www.springer.com/lncs}}

%\toctitle{Lecture Notes in Computer Science}
%\tocauthor{Authors' Instructions}
\clearpage
\maketitle
\thispagestyle{empty}

\begin{abstract}
We study algorithms based on local improvements for the $k$-Set Packing problem. The well-known local improvement algorithm by Hurkens and Schrijver \cite{schrijver} has been improved by Sviridenko and Ward \cite{ksetpacking2013} from $\frac{k}{2}+\epsilon$ to $\frac{k+2}{3}$, and by Cygan \cite{bestksetpacking} to $\frac{k+1}{3}+\epsilon$ for any $\epsilon>0$.
%The local improvements considered in \cite{ksetpacking2013,bestksetpacking} are of logarithmic size.

In this paper, we achieve the approximation ratio $\frac{k+1}{3}+\epsilon$ for the $k$-Set Packing problem using a simple polynomial-time algorithm based on the method by Sviridenko and Ward \cite{ksetpacking2013}. With the same approximation guarantee, our algorithm runs in time singly exponential in $\frac{1}{\epsilon^2}$, while the running time of Cygan's algorithm \cite{bestksetpacking} is doubly exponential in $\frac{1}{\epsilon}$. On the other hand, we construct an instance with locality gap $\frac{k+1}{3}$ for any algorithm using local improvements of size $O(n^{1/5})$, here $n$ is the total number of sets. Thus, our approximation guarantee is optimal with respect to results achievable by algorithms based on local improvements.


\end{abstract}

%\newpage

\section{Introduction}

Given a universe of elements $U$ and a collection $\mathcal{S}$ of subsets with size at most $k$ of $U$, the $k$-Set Packing problem asks to find a maximum number of disjoint sets from $\mathcal{S}$. The most prominent approach for the $k$-Set Packing problem is based on local improvements. In each round, the algorithm selects $p$ sets from the current packing and replaces them with $p+1$ sets such that the new solution is still a valid packing. It is well-known that for any $\epsilon>0$, there exists a constant $p$, such that the local improvement algorithm has an approximation ratio $\frac{k}{2}+\epsilon$ \cite{schrijver}. In quasi-polynomial time, the result has been improved to $\frac{k+2}{3}$ \cite{Halldorsson1995} and later to $\frac{k+1}{3}+\epsilon$ for any $\epsilon>0$ \cite{sellhyperedge} using local improvements of size $O(\log n)$, here $n$ is the size of $\mathcal{S}$. In \cite{sellhyperedge}, the algorithm looks for any local improvement of size $O(\log n)$, while in \cite{Halldorsson1995}, only sets which intersect with at most 2 sets in the current solution are considered and the algorithm looks for improvements of a binocular shape.

One can obtain a polynomial-time algorithm which looks for local improvements of logarithmic size using the color coding technique \cite{ksetpacking2013,bestksetpacking}. The algorithm in \cite{ksetpacking2013} looks for local improvements similar to \cite{Halldorsson1995} and has an approximation ratio $\frac{k+2}{3}$. In \cite{bestksetpacking}, local improvements of bounded pathwidth are considered and an approximation ratio $\frac{k+1}{3}+\epsilon$, for any $\epsilon>0$ is achieved.

%The algorithm of \cite{ksetpacking2013} has a better performance guarantee than \cite{schrijver} in all $k$'s except for $k=3$. We show that an algorithm looking for both types of local improvements in \cite{schrijver} and \cite{ksetpacking2013} has an improved approximation ratio of $\frac{k^2-2}{3k-4}+\epsilon$, which is at most $\frac{k+4/3}{3}$. This result is better than both \cite{schrijver} and \cite{ksetpacking2013} for any $k$. We use the factor-revealing linear program introduced in \cite{factorreveal} for the analysis. We also observe that the result is tight for this algorithm.

In this paper, we obtain an approximation ratio $\frac{k+1}{3}+\epsilon$ for the $k$-Set Packing problem, for any $\epsilon>0$. On the other hand, we improve the lower bound given in \cite{ksetpacking2013} by constructing an instance that any algorithm using local improvements of size $O(n^{1/5})$ has a performance ratio at least $\frac{k+1}{3}$. Thus, our result is optimal with respect to the performance guarantee achievable by algorithms using local improvements. Our algorithm extends the types of local improvements considered in \cite{Halldorsson1995,ksetpacking2013} by first looking for a series of set replacements which swap some sets in the current packing $\mathcal{A}$ with a same number of disjoint sets $\mathcal{T}$ which are not in $\mathcal{A}$. We then look for local improvements which can be decomposed into cycles and paths, from sets in $\mathcal{S}\setminus(\mathcal{A}\cup\mathcal{T})$ which intersect with at most 2 sets in $\mathcal{A}$. We also use the color-coding technique \cite{colorcoding,fixparapacking} to ensure a polynomial time complexity when the local improvement has logarithmic size.
Our algorithm is more efficient as it runs in time singly exponential in $\frac{1}{\epsilon^2}$, while the running time of Cygan's algorithm \cite{bestksetpacking} is doubly exponential in $\frac{1}{\epsilon}$. We believe that our approach makes an important step towards a practical algorithm for the $k$-Set Packing problem.

{\bf Related works.} The Set Packing problem has been studied for decades. Hastad has shown that the general Set Packing problem cannot be approximated within $N^{1-\epsilon}$ unless $NP\subseteq ZPP$ \cite{setpackinglowerbound}. Here $N$ is the size of the universe $U$. The bounded Set Packing problem assumes an upper bound of the size of the sets. In the unweighted case, i.e. the $k$-Set Packing problem, besides algorithms based on local improvements \cite{schrijver,Halldorsson1995,sellhyperedge,ksetpacking2013,bestksetpacking}, Chan and Lau have shown that the standard linear programming algorithm has an integrality gap $k-1+1/k$ \cite{lpsdppacking}. They have also constructed a polynomial-sized semi-definite program with integrality gap $\frac{k+1}{2}$, but no rounding strategy is provided. The problem is also known to have a lower bound $\Omega(\frac{k}{\log k})$ \cite{ksetpackinglowerbound}. In the weighted case, Chandra and Halld\'{o}rsson have given a nontrivial approximation ratio $\frac{2(k+1)}{3}$ \cite{weightpacking2}. The result was improved to $\frac{k+1}{2}+\epsilon$ by Berman \cite{weightpacking}, which remains the best so far.

%They also show that any algorithm using local improvements of up to linear size has an approximation ratio at most $\frac{k}{3}$.

The paper is organized as follows. In Section 2, we review previous local search algorithms and define useful tools for analysis. In Section 3, we introduce the new local improvement and analyze its performance guarantee. In Section 4, we give an efficient implementation of our algorithm. In Section 5, we give a lower bound of algorithms based on local improvements for the $k$-Set Packing problem. We conclude in Section 6.

\section{Preliminaries}

\subsection{Local improvements}

Let $\mathcal{S}$ be a collection of subsets of size at most $k$ of the universe $U$ and the size of $\mathcal{S}$ is $n$. Let $\mathcal{A}$ be the collection of disjoint sets chosen by the algorithm. In this paper, we are interested in the unweighted $k$-set packing problem. We assume without loss of generality that every set is of uniform size $k$. Otherwise, we could add distinct elements to any set until it is of size $k$. In the following context, we use calligraphic letters to represent collections of $k$-sets and capital letters to represent sets of vertices which correspond to $k$-sets.

The most widely used algorithm for the $k$-Set Packing problem is local search. The algorithm starts by picking an arbitrary maximal packing. If there exists a collection of $p+1$ sets $\mathcal{P}$ which are not in $\mathcal{A}$ and a collection of $p$ sets $\mathcal{Q}$ in $\mathcal{A}$, such that $(\mathcal{A}\setminus \mathcal{Q})\cup \mathcal{P}$ is a valid packing, the algorithm will replace $\mathcal{Q}$ with $\mathcal{P}$. We call it a {\it $(p+1)$-improvement}.

With $p$ being a constant which depends on $\epsilon$, it is well-known that this local search algorithm achieves an approximation ratio $\frac{k}{2}+\epsilon$, for any $\epsilon>0$ \cite{schrijver}.

\begin{theorem}[\cite{schrijver}]
\label{thm:hs}
For any $\epsilon>0$, there exists an integer $p=O(\log_k \frac{1}{\epsilon})$, such that the local search algorithm which looks for any local improvement of size $p$ has an approximation ratio $\frac{k}{2}+\epsilon$ for the $k$-Set Packing problem.
\end{theorem}

Halld\'{o}rsson \cite{Halldorsson1995} and later Cygan et al. \cite{sellhyperedge} show that when $p$ is $O(\log n)$, the approximation ratio can be improved at a cost of quasi-polynomial time complexity. Based on the methods of \cite{Halldorsson1995}, Sviridenko and Ward \cite{ksetpacking2013} have obtained a polynomial-time algorithm using the color coding technique \cite{colorcoding}.
We summarize their algorithm as follows. Let $\mathcal{A}$ be the packing chosen by the algorithm and $\mathcal{C}=\mathcal{S}\setminus \mathcal{A}$. Construct an {\it auxiliary multi-graph $G_{A}$} as follows. The vertices in $G_{A}$ represent sets in $\mathcal{A}$. For any set in $\mathcal{C}$ which intersects with exactly two sets $s_1,s_2\in\mathcal{A}$, add an edge between $s_1$ and $s_2$. For any set in $\mathcal{C}$ which intersects with only one set $s\in \mathcal{A}$, add a self-loop on $s$. The algorithm searches for local improvements which can be viewed as {\it binoculars} in $G_{A}$. They call them {\it canonical improvements} \cite{ksetpacking2013}. A binocular can be decomposed into paths and cycles. The color coding technique \cite{colorcoding} and the dynamic programming algorithm are employed to efficiently locate paths and cycles of logarithmic size. This algorithm has an approximation ratio at most $\frac{k+2}{3}$.

\begin{theorem}[\cite{ksetpacking2013}]
\label{thm:sw}
With $p=4\log n+1$, there exists a polynomial-time algorithm which solves the $k$-Set Packing problem with an approximation ratio $\frac{k+2}{3}$.
\end{theorem}

\begin{comment}
We remark that this approximation ratio is tight. A tight example is as follows. Assume $\mathcal{C}=\mathcal{S}\setminus \mathcal{A}$ can be partitioned into two collections, one collection $\mathcal{C}_1$ of $|\mathcal{A}|$ sets where every set intersects with 1 set in $\mathcal{A}$, and the other collection $\mathcal{C}_3$ of $\frac{k-1}{3}|\mathcal{A}|$ sets where every set intersects with 3 sets in $\mathcal{A}$. It is easy to see that there is no canonical improvement.

\begin{figure}[!t]
\centering
\label{canonicalimprovement}
\includegraphics[width=4in]{canonicalimprovement6.pdf}
\caption{Canonical improvements. The three types in the first row are considered in \cite{ksetpacking2013}.}
\end{figure}
\end{comment}

Cygan \cite{bestksetpacking} has shown that an approximation ratio $\frac{k+1}{3}+\epsilon$ can be obtained in polynomial time by restricting the local improvements from anything of size $O(\log n)$ \cite{sellhyperedge} to local improvements of bounded pathwidth. Namely, let $G(A,C)$ be the {\it bipartite conflict graph} where $\mathcal{A}$ and $\mathcal{C}=\mathcal{S}\setminus\mathcal{A}$ represent one part of vertices respectively. For any $u\in A, v\in C$, if the corresponding sets are not disjoint, we put an edge between $u$ and $v$. For any disjoint collection $P\subseteq C$, if the subgraph induced by $P$ and the neighbors of $P$, $N(P)$ in $A$ have bounded pathwidth, a set replacement of $P$ with $N(P)$ is called a local improvement of bounded pathwidth. The color coding technique is also employed for efficiently locating such an improvement.

\begin{theorem}[\cite{bestksetpacking}]
For any $\epsilon> 0$, there exists a local search algorithm which runs in time $2^{O(kr)}n^{O(pw)}$ with an approximation ratio $\frac{k+1}{3}+\epsilon$ of the k-Set Packing problem. Here $r=2(k+1)^{\frac{1}{\epsilon}}\log n$ is the upper bound of the size of a local improvement, $pw=2(k+1)^{\frac{1}{\epsilon}}$ is the upper bound of pathwidth.
\end{theorem}

\subsection{Partitioning the bipartite conflict graph}

Consider a bipartite conflict graph $G(A,B)$ where one part of the vertices $A$ representing sets $\mathcal{A}$ chosen by the algorithm and the other part $B$ representing an arbitrary disjoint collection of sets $\mathcal{B}$. We assume without loss of generality that $\mathcal{B}\cap \mathcal{A}=\emptyset$. The collection $\mathcal{B}$ can be thought of an optimal solution. It is only used for analysis.

Given $\epsilon>0$, let $c_k=k-1$, $b=|B|$, we further partition $G(A, B)$ iteratively as follows. Let $B_1^1$ be the set of vertices in $B$ with degree 1 to $A$. Denote the neighbors of $B_1^1$ in $A$ by $A_1^1$. If $|B_1^1|<\epsilon b$, stop the partitioning. Otherwise, we consider $B_1^2$ which is the set of vertices whose degree drops to 1 if we remove $A^1_1$. Denote the neighbors of $B_1^2$ in $A\setminus A^1_1$ by $A_1^2$. If $|B_1^1\cup B_1^2|< c_k\epsilon b$, stop the partitioning. In general for any $j\geq 2$, let $B_1^j$ be the set of vertices with their degree dropping to 1 if the vertices in $\cup_{l=1}^{j-1}A_1^l$ are removed, and let $A_1^j$ be the neighbors of $B_1^j$ which are not in $\cup_{l=1}^{j-1}A^l_1$. If $|\cup_{l=1}^j B_1^l|<c_k^{j-1}\epsilon b$, we stop. Otherwise continue the partitioning. Let $i$ be the smallest integer such that $|\cup_{l=1}^i B_1^l|<c_k^{i-1}\epsilon b$. This integer $i$ exists as $c_k^{i-2}\epsilon b\leq |\cup_{l=1}^{i-1} B_1^l|\leq b$, we have $i\leq 2+\log_{c_k}\frac{1}{\epsilon}$.
Let $B_1^{\leq j}$ ($A_1^{\leq j}$) be the set union $\cup_{l=1}^j B_1^l$ ($\cup_{l=1}^j A_1^l$), for $j\geq 1$.
%We give an illustration of the partition in Figure 1.
%Let $B_1= \cup_{j=1}^{i} B_1^i$ and $A_1=A_1\cup A_1^i$.

\begin{comment}
Furthermore, let $B_2$ be the collection of vertices of degree 2 to $A$ and with no neighbors in $A_1$. Let $B_3^2$ be the collection of vertices with degree at least 3 to $A$ and degree 2 to $A\setminus A_1$. Let $B_3^3$ be the collection of sets with degree at least 3 to $A$ and degree at least 3 to $A\setminus A_1$.

If we first run Algorithm LI1 with parameter $p=O(\frac{k}{\epsilon})$ to ensure that there is no local improvement of size $p$, there will be no vertex $b$ in $B$ such that after removing some $B_1^j$, its degree drops to 0. Otherwise, suppose $b$ has all its neighbors $N(b)\subseteq A_1$ and $N(b)$ have neighbors $N(N(b))\subseteq B_1$. We have $|N(N(b))|>|N(b)|$. Hence, there exists a local improvement by replacing $N(b)$ with $N(N(b))$. Therefore, $B_1, B_2, B_3^2, B_3^3$ are all possible types of vertices in $B$ and they are mutually disjoint. Hence, $B=B_1\cup B_2\cup B_3^2 \cup B_3^3$. We denote the size of $B_1,B_2,B_3^2,B_3^3$ by $b_1,b_2,b_3^2,b_3^3$ respectively. An illustration of the definitions is given in Figure 2.
\end{comment}

\begin{comment}
\begin{figure}[!t]
\centering
\label{auxiliarygraph}
\includegraphics[width=3.5in]{auxiliarygraph1.pdf}
\caption{The bipartite conflict graph.}
\end{figure}
\end{comment}

\section{Canonical Improvements with Tail Changes}

\begin{comment}
Before introducing the new local search algorithm, we show that the performance guarantee of the $k$-Set Packing problem can be improved to $\frac{k^2-2}{3k-4}+\epsilon$ by using both local improvements
in \cite{ksetpacking2013} and \cite{schrijver} with parameter $p=O(\frac{k}{\epsilon})$. This new algorithm ({\bf Algorithm LI2*}) iterates by first running Algorithm LI1 until
there is no local improvements of constant size available. Then it switches to Algorithm LI2 and searches for canonical improvements of $O(\log n)$ size.

\begin{theorem}
\label{thm:ls21}
For any $\epsilon >0$, Algorithm LI2* has an approximation ratio $\frac{k^2-2}{3k-4}+\epsilon$. In particular, the 3-Set Packing problem has an approximation ratio $1.4+\epsilon$.
\end{theorem}

We include a proof of Theorem \ref{thm:ls21} in Appendix A.
\end{comment}

In this section, we present a local search algorithm based on \cite{ksetpacking2013}, and show that it achieves an approximation ratio $\frac{k+1}{3}+\epsilon$ for any $\epsilon>0$ for the $k$-Set Packing problem.

\subsection{The new local improvement}

In this section, we introduce a new type of local improvements. Let $\mathcal{A}$ be a packing chosen by the algorithm, and let $\mathcal{C}=\mathcal{S}\setminus \mathcal{A}$. We create the bipartite conflict graph $G(A,C)$ as in Section 2.1. Recall that only those sets in $\mathcal{C}$ which intersect with at most 2 sets in $\mathcal{A}$ are considered in \cite{ksetpacking2013}. Our approach tries to include sets of higher degree in a local improvement by swapping $p$ sets in $\mathcal{A}$ with $p$ sets in $\mathcal{C}$. In this way, if the degree of a vertex in $C$ drops to 2, it could be included in a local improvement.

\begin{definition}[Tail change]
Consider any vertex $v\in C$ of degree at least 3, we call a swapping of $p$ sets $U$ in $A$ with $p$ disjoint sets $V$ in $C$ a {\bf tail change} associated with an edge $(v, u)$ of $v$ if the following three requirements are satisfied: (1) $v\notin V$. (2) $u$ is the unique neighbor of $v$ in $U$. (3) The neighbors of $V$ in $A$ are exactly $U$.
The {\bf size} of this tail change is defined to be $p$.
\end{definition}

We denote a tail change associated with $e$ which swaps $U$ with $V$ by $T_e(U,V)$.
We say that two tail changes $T_e(U,V),T_{e'}(U',V')$ of vertices $v,v'$ respectively are {\it consistent}, if either $v\neq v'$ and $(\{v\}\cup V)\cap (\{v'\}\cup V')=\emptyset$, or $v=v'$, $e\neq e'$ and $V\cap V'=\emptyset$. Moreover we require that the degrees of $v,v'$ after the tail changes remain at least 2. Therefore, for any vertex $v\in C$ of degree $d\geq 3$, we could perform at most $d-2$ tail changes for $v$.

We are now ready to introduce the new local search algorithm. We first consider an algorithm that runs in quasi-polynomial time. Given parameter $\epsilon>0$, in each iteration, the algorithm starts by performing local improvements of constant size $O(\frac{k}{\epsilon})$. If no such local improvement is present, the algorithm starts looking for improvements of size $O(\log n)$. Construct the bipartite conflict graph $G(A,C)$. For any set $I$ of at most $\frac{4}{\epsilon}\log n$ vertices in $C$, let $I_3\subseteq I$ be the set of vertices of degree at least 3 in $G(A,C)$. The algorithm checks if there exists a collection of consistent tail changes each of size at most $\frac{2(k-1)}{\epsilon}$ for $I_3$, which together replace $U\subseteq A$ with $V\subseteq C$, such that $V\cap I=\emptyset$, and after the replacement the degree of every vertex in $I_3$ drops to 2. If so, the algorithm goes on checking in the auxiliary multi-graph $G_A$ where edges are constructed from vertices in $I$ assuming the swapping of $U$ with $V$ is performed, whether there is a subgraph which is one of the following six types (illustrated in Figure 1): (1) two cycles intersecting at a single point, (2) two disjoint cycles connecting by a path, (3) two cycles with a common arc, (those three types are binoculars also considered in \cite{ksetpacking2013}), (4) a path, (5) a path and a cycle intersecting at a single point, (6) a cycle. Let $U'$ be the vertices in this subgraph, and $V'$ be the edges. The algorithm checks if a replacement of $U\cup U'$ with $V\cup V'$ is an improvement. We call this new local improvement the {\bf canonical improvement with tail changes}, and this quasi-polynomial time algorithm, {\bf Algorithm LI} (LI stands for local improvement). We will explain the parameter settings in Algorithm LI in the next section.

\begin{figure}[!t]
\centering
\label{canonicalimprovement}
\includegraphics[width=4in]{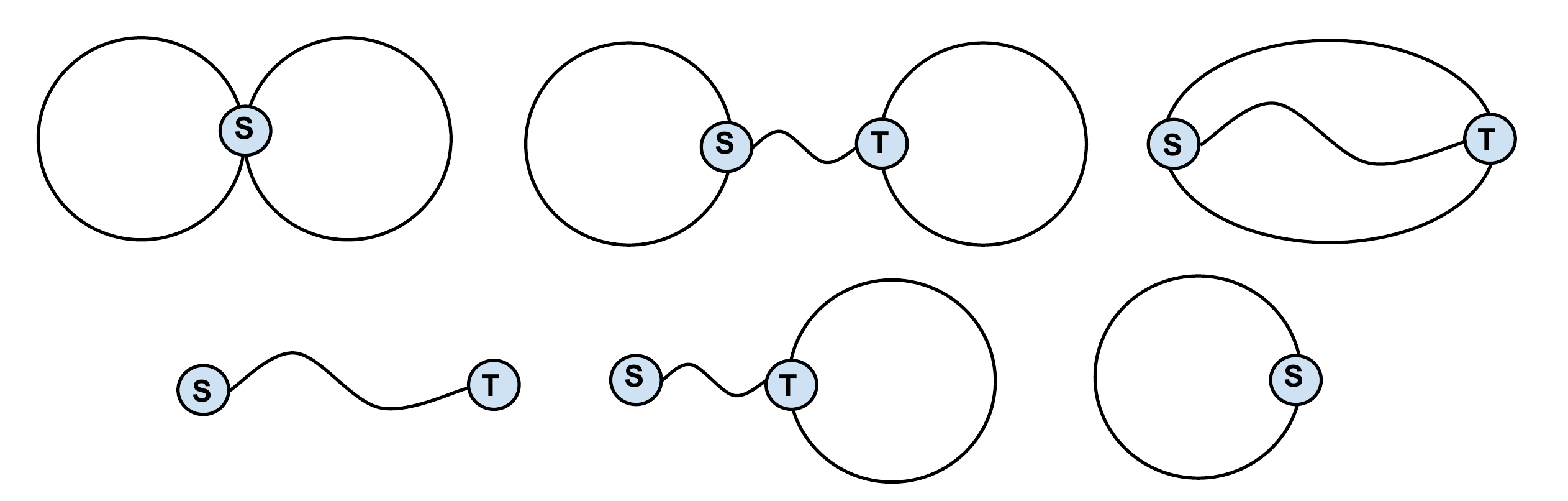}
\caption{Improvements composed of cycles and lines.}
\end{figure}

Before showing how to efficiently locate a canonical improvement with tail changes, we first show that the approximation ratio of Algorithm LI is $\frac{k+1}{3}+\epsilon$.

\subsection{Analysis}

Given a packing $A$ chosen by Algorithm LI and for an arbitrary packing $B$, consider the bipartite conflict graph $G(A,B)$ defined in Section 2.2. The notations in this section are taken from Section 2.2. First, we remark that since we make all $O(\frac{k}{\epsilon})$-improvements at the beginning of Algorithm LI, for any set $V\subseteq B$ of size $O(\frac{k}{\epsilon})$, there are at least $|V|$ neighbors of $V$ in $A$.
In $G(A, B)$, we make every vertex $a$ in $A$ full degree $k$ by adding self-loops of $a$ which we call {\it null edges}. We define a {\it surplus edge} which is either a null edge, or an edge incident to some vertex in $B$ which is of degree at least 3. We first show that there exists a one-to-one matching from almost all vertices in $B_1^1$ to surplus edges with the condition that after excluding the matched surplus edges of any vertex in $B$, the degree of this vertex remains at least 2. We define such a matching in the following matching process.

%{\bf The matching process.} Starting from a vertex $v'\in B_1^{1}$, go to its only neighbor $u'\in A_1$. If the degree of $u'$ is at most 2, match $v'$ to a surplus degree of $u'$, decrease $\widetilde{d}_{u'}$ by 1 and stop. Otherwise, if $u'$ has a neighbor $v$ in $B$ such that $v$ is unmatched and the degree of $v$ is at least 3, match $v'$ to one unmatched edge of $v$. Mark this edge as matched. If the degree of $v$ drops to 2 by excluding all matched edges, mark $v$ as matched. If $u'$ does not have a neighbor satisfying the requirement, try every neighbor $v_1$ of $u$ in $B_1$ and continue the process from $v_1$ until an unmatched edge is found. If $u$ does not have any neighbor in $B_1$, the matching process fails.

{\bf The matching process.} Pick an arbitrary order of vertices in $B_1^1$. Mark all edges and vertices as unmatched. Try to match every vertex with a surplus edge in this order one by one. For any vertex $v_1\in B_1^{1}$, starting from $v_1$, go to its neighbor $u_1\in A_1^1$. If $u_1$ has an unmatched null edge, match $v_1$ to it, mark this null edge as matched and stop. Otherwise, if $u_1$ has a neighbor $v$ in $B$, such that the degree of $v$ is at least 3, $(u_1,v)$ is unmatched and $v$ is unmatched, match $v_1$ to $(u_1,v)$ and mark this edge as matched. If the degree of $v$ drops to 2 by excluding all matched edges of $v$, mark $v$ as matched. If $u_1$ does not have a neighbor satisfying the requirement, try every neighbor $v_2$ (except $v_1$) of $u_1$ and continue the process from $v_2$. In general, suppose we are at a vertex $v_j\in B_1^j$ and it has a neighbor $u_j\in A_1^j$. We try to match $v_1$ with a null edge of $u_j$, or a surplus edge of an unmatched neighbor of $u_j$. If no matching edge is found, continue by trying every neighbor of $u_j$ in $B_1^{j_1}$ for $j_1>j$, until either $v_1$ is matched, or $j> 2+\log_{c_k}\frac{1}{\epsilon}$. In the latter case, we mark $v_1$ as unmatched.

\begin{figure}[!t]
\centering
\label{conflictgraph}
\includegraphics[width=3.5in]{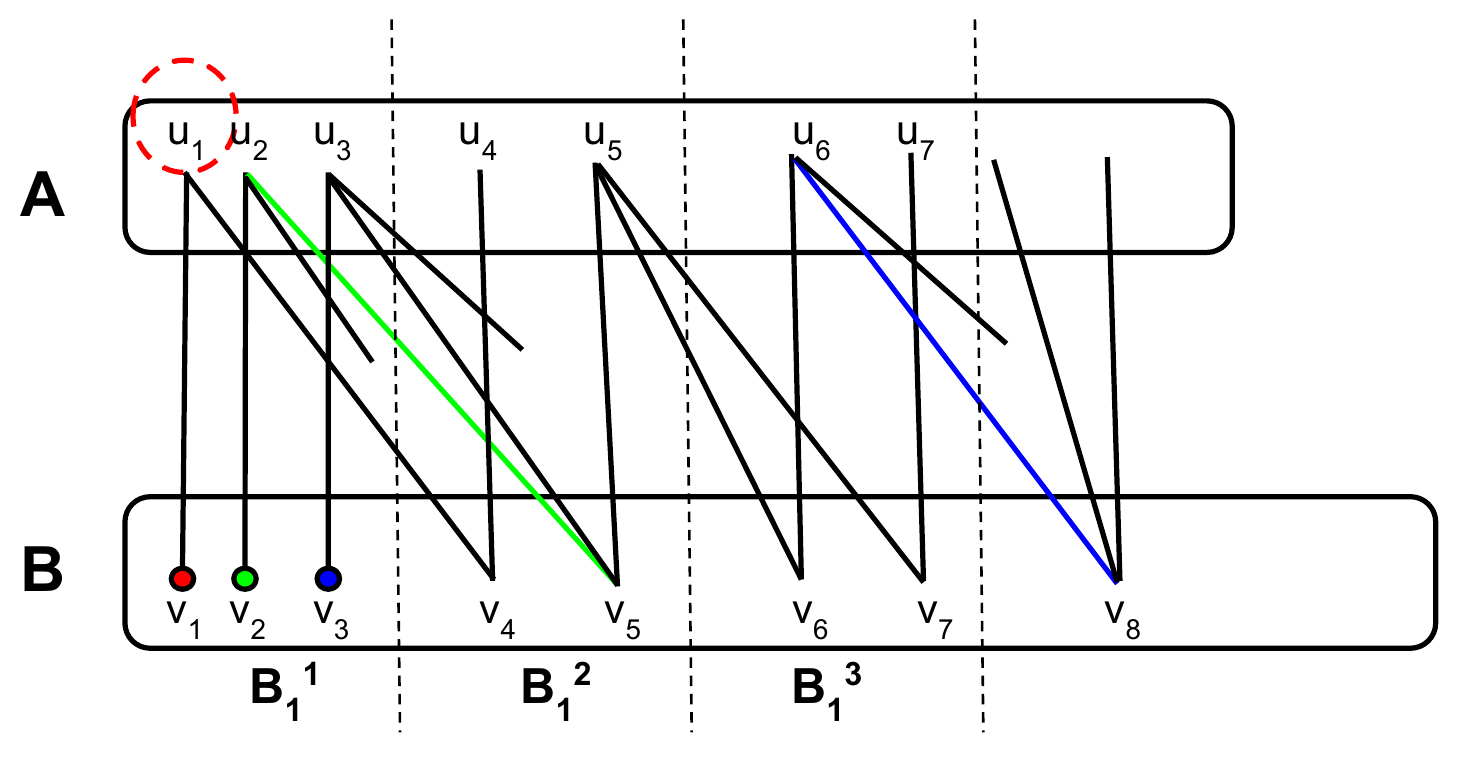}
\caption{The bipartite conflict graph. $k=3$.}
\end{figure}

We give an example of the matching process illustrated in Figure 2. The vertices and edges of the same color are matched. We match $v_1$ to the null edge (dotted line) of its neighbor $u_1$. $v_2$ is matched to a surplus edge $(u_2,v_5)$ of $v_5$. After that, the degree of $v_5$ drops to 2 by excluding the edge $(u_2,v_5)$ and $v_5$ is marked as matched. For $v_3$, we go on to $u_3$, $v_5$, $u_5$, $v_6$, $u_6$, then $v_8$ with a surplus edge $(u_6,v_8)$. We match $v_3$ to this edge.

\begin{comment}
\begin{figure}[!t]
\centering
\label{figmatch}
\includegraphics[width=3.5in]{tailchange2.pdf}
\caption{An example of the matching process for the 3-Set Packing problem. The self-loop in dotted line represents a null edge. Vertices and edges in the same color are matched. }
\end{figure}
\end{comment}

\begin{lemma}
\label{lemma:matchb1}
For any $\epsilon>0$, there exists a set of surplus edges $E_1$, such that except for at most $\epsilon |B|$ vertices, $B_1^{1}$ can be matched to $E_1$ one-to-one. Moreover, every endpoint of $E_1$ in $B$ has degree at least 2 after excluding $E_1$.
\end{lemma}

\begin{proof}

It is sufficient to prove that at most $\epsilon |B|$ vertices in $B_1^1$ are unmatched. Let $v$ be an unmatched vertex in $B_1^1$. The neighbor of $v$ in $A$, $u$ has no null edges and thus has degree $k$, and none of the neighbors of $u$ have an unmatched surplus edge. The matching process runs in exactly $i-1=1+\log_{c_k}\frac{1}{\epsilon}$ iterations for $v$ and has tried $k-1+(k-1)^2+\cdots +(k-1)^{i-1}= \frac{(k-1)^i-(k-1)}{k-2}$ vertices in $B_1^{\leq i}$. Notice that for two different vertices $v,v'\in B_1^1$ which are marked unmatched, the set of vertices the matching process has tried in $B_1^{\leq i}$ must be disjoint. Otherwise, we can either find a matching for one of them, or there exists a local improvement of size $O(\frac{k}{\epsilon})$. Suppose there are $n_{um}$ unmatched vertices in $B_1^1$. Recall that $|B_1^{\leq i}|\leq c_k^{i-1}\epsilon |B|$. Therefore, $n_{um}\dot (1+\frac{(k-1)^i-(k-1)}{k-2})\leq |B_1^{\leq i}|\leq c_k^{i-1}\epsilon |B|$. We have $n_{um}\leq \epsilon |B|$, where $c_k= k-1$. $\square$
\end{proof}

%If we view the matching process reversely, we can formalize Lemma \ref{lemma:matchb1} from the perspective of tail changes. Consider a vertex $v^*\in B_1^1$ which is matched to a surplus edge $e=(u,v)$ after trying the set of vertices $V'\subseteq B$. Let $U_1$ be the neighbors of $V'$ in $A_1$ and $V_1$ be the neighbors of $U_1$ in $B_1^1$. We have a special tail change $T^*_e(U,V,j)$ associated with $e$, where $V=V'\cup V_1$, $U$ is the set of neighbors of $V$, and $j\leq 2+\log_{c_k}\frac{1}{\epsilon}$.

Consider any surplus edge $e=(u,v)$ matching to $w\in B_1^1$, we can obtain a tail change $T_e(U,V)$ associated with $e$ by viewing the matching process reversely. Assume $u\in A_1^i$, for $i<2+\log_{c_k}\frac{1}{\epsilon}$. Let $U_i=\{u\}$ and $V_i$ be the neighbor of $U_i$ in $B_1^i$. In general, let $V_j$ be the set of neighbors of $U_j$ in $B_1^j$, we define $U_{j-1}$ to be the neighbors of $V_j$ excluding $U_j$. This decomposition ends when $j=1$ and $V_1\subseteq B_1^1$. As $w$ is matched to $e$, we know that $w\in V_1$. Let $U=\cup_{j=1}^iU_j,V=\cup_{j=1}^i V_j$, then a swapping of $U$ with $V$ is a tail change associated with edge $e$. First, we have $|V_j|=|U_j|$ for $1\leq j\leq i$, otherwise there exists a $O(\frac{k}{\epsilon})$-improvement. Hence $|U|=|V|$. Secondly, the set of neighbors of $V$ is $U$ by the construction. And $u$ is the only neighbor of $v$ in $U$, otherwise $w$ will be matched to another surplus edge. As an example in Figure 2, $U=(u_6,u_5,u_3,u_2),V=(v_6,v_5,v_3,v_2)$ is a tail change associated with edge $(u_6,v_8)$ which is matched to $v_3$.

We estimate the size of such a tail change. Since every vertex in $U$ has at most $k$ neighbors, there are at most $\sum_{j=1}^{i} (k-1)^{j-1}=\frac{(k-1)^i-1}{k-2}$ vertices in $V$. Let $i=2+\log_{k-1}\frac{1}{\epsilon}$. Then the size of this tail change is at most $\frac{2(k-1)}{\epsilon}$.

\begin{comment}
If we view the matching process reversely, we can decompose a tail change $T_e(U,V)$ with $e=(v,u)$, $U\subseteq A, V\subseteq B$ as follows. The decomposition is similar as in Section 2.2.
We assume $V_1=V\cap B_1^1$ is not empty. Let $U_1$ be the neighbors of $V_1$. In general, let $V_j$ be the set of vertices in $V$ with degree 1 to $A\setminus \cup_{l=1}^{j-1}U_l$, and let $U_j$ be the neighbors of $V_j$ excluding those in $\cup_{l=1}^{j-1}U_l$. We call a tail change a {\it special tail change} if it has the following properties: (1) $V_1\neq\emptyset$. (2) Every vertex in $U_j$ has a neighbor in $V_{j+1}$. (3) There exists some integer $i\leq 1+\log_{c_k}\frac{1}{\epsilon}$, such that the decomposition ends at $U_i$, and $u$ is the only vertex in $U_i$. (4) $V=\cup_{l=1}^{i}V_l$, $U=\cup_{l=1}^{i}U_l$. (5) $v\in V_m$ for $i<m\leq 2+\log_{c_k}\frac{1}{\epsilon}$ with $u$ being its only neighbor in $U$. We denote a special tail change by $T_e^*(U,V,i)$.
\end{comment}

Let $\mathcal{T}_S$ be the collection of tail changes associated with the surplus edges which are not null edges as defined above. Assume those tail changes together replace $U\subseteq A$ with $V\subseteq B$. Let $B_L=B_1^1\cap V$. Let $B_N\subseteq B_1^1$ be the set of vertices which are matched to null edges. By Lemma \ref{lemma:matchb1}, we know that $|B_1^1\setminus (B_L\cup B_N)|\leq \epsilon |B|$. Moreover, we show in the following Corollary that a consistent collection of tail changes with the same property can be extracted from $\mathcal{T}_S$.

%Notice that by definition of the matching process, for any $T_{e_1},T_{e_2}\in\{T_e\}$, either they are consistent, or one of them is included in the other. Therefore, $\{T_e\}$ is a partially ordered sets by inclusion. We consider all maximal elements of $\{T_e\}$.

\begin{corollary}
\label{corofmatchprocess}
There exists a subcollection $\mathcal{T}_c$ of consistent tail changes from $\mathcal{T}_S$, which together replace a set $U_c\subseteq A$ with $V_c\subseteq B$, such that $V_c\cap B_1^1=B_L$.
\end{corollary}

\begin{proof}
We consider the tail changes associated with the surplus edges one by one. $\mathcal{T}_c$ is initialized to be empty. If the tail change $T_{e_i}(U_i,V_i)$ is consistent with every tail change in $\mathcal{T}_c$, we include it in $\mathcal{T}_c$. If there exists any tail change $T_{e_j}(U_j,V_j)$ such that $V_i\cap V_j \neq \emptyset$, assume $e_j=(u_j,v_j)$ is matched with the vertex $w_j$ and $e_i=(u_i,v_i)$ with $w_i$, we know that at the time the matching for $w_i$ tries the edges of $v_j$, $v_j$ has been marked as matched. Hence, $V_j\subseteq V_i$. We discard $T_{e_j}(U_j,V_j)$ and include $T_{e_j}(U_j,V_j)$ in $\mathcal{T}_c$. $\square$

\end{proof}

\begin{theorem}
\label{thm:ls3}
For any $\epsilon>0$, Algorithm LI has an approximation ratio $\frac{k+1}{3}+\epsilon$.
\end{theorem}

Before proving the theorem, we state the following result from \cite{ksetpacking2013} which is derived from a lemma in \cite{indsetbounddeg}. The lemma in \cite{indsetbounddeg} states that when the density of a graph is greater than a constant $c>1$, there exists a subgraph of size $O(\log n)$ with more edges than vertices. If the underlying graph is the auxiliary multi-graph $G_A$ defined in Section 2.1 and this condition holds, we know from \cite{ksetpacking2013} that there exists a binocular of size $O(\log n)$.

\begin{lemma}[\cite{ksetpacking2013}]
\label{lemmacycle}
For any integer $s\geq 1$ and any undirected multigraph $G=(V,E)$ with $|E|\geq \frac{s+1}{s}|V|$, there exists a binocular of size at most $4s\log n-1$.
\end{lemma}

\begin{proof}[Theorem \ref{thm:ls3}]
For a given $\epsilon$, let $\epsilon' = \frac{2k+5}{3}\epsilon>3\epsilon$. Let $A$ be the packing returned by Algorithm LI with parameter $\epsilon$ and for any other packing $B$, we show that $(\frac{k+1}{3}+\epsilon)|A|\geq |B|$.
In the following, we use the corresponding small letter of a capital letter (which represents a set) to represent the size of this set.

In Corollary \ref{corofmatchprocess}, the collection of consistent tail changes $\mathcal{T}_c$ together replace a set of $a_t$ vertices $A_t$ in $A$ with $b_t=a_t$ vertices $B_t$ in $B$. We exclude theses vertices from the original bipartite conflict graph $G(A, B)$. Denote the remaining graph by $G(A',B')$. We add null edges to vertices in $A'$ until every vertex in $A'$ has degree $k$. There are $ka'$ edges counting from $A'$.

Let $B'_N= B_N\cap B'$ and $b^1_n=|B_N'|$. We can also think of that there is a null edge at each vertex in $B_N'$ when we count the number of edges from $B'$. By Lemma \ref{lemma:matchb1}, there are at most $\epsilon' b$ unmatched vertices in $B_1^1$. We further partition the vertices in $B'$ as follows. Let $B_3^2$ ($B_3^{2'}$) be the set of vertices in $B'$ whose degree drops to 2 after performing the tail changes in $\mathcal{T}_c$, and for any vertex $v\in B_3^2$ ($v'\in B_3^{2'}$), there is at least one (no) tail change in $\mathcal{T}_c$ associated with $v$.
Let $B^2_2$ be the set of vertices in $B'$ with degree 2 in $G(A',B')$ and no neighbors in $A_t$. Let $B_2^1$ ($B_3^1$) be the set of vertices in $B'$ whose degree drops from 2 (at least 3) in $G(A,B)$ to 1 in $G(A',B')$. Let $B_3^3$ be the set of vertices in $B'$ with degree at least 3 in $G(A',B')$.
Moreover, there is no vertex in $G(A',B')$ of degree 0, otherwise, there exists a local improvement.

By Lemma \ref{lemma:matchb1}, the number of edges in $G(A', B')$ is at least $2b^1_n+\epsilon' b+b_2^1+2b^2_2+2b_3^2+3b_3^3+2b_3^{2'}+b_3^1$. Therefore,
\begin{equation}
\label{relationab1}
k(a-a_t) \geq 2b^1_n + \epsilon' b+b_2^1+2b^2_2+2b_3^2+3b_3^3+2b_3^{2'}+b_3^1.
\end{equation}

Next, we show that $b_n^1+b_2^1+b^2_2+b_3^2\leq (1+\epsilon')(a-a_t)$. Suppose the set of neighbors of $B_N',B^2_2,B_2^1,B_3^2$ in $A'$ is $A_2$. Construct an auxiliary multi-graph $G_{A_2}$ as in Section 2.1, where the vertices are $A_2$, every vertex in $B_N',B_2^1$ creates a self-loop, and every vertex in $B_2^2,B_3^2$ creates an edge. Assume $G_{A_2}$ has at least $(1+\epsilon')|A_2|$ edges, implied by Lemma \ref{lemmacycle}, there exists a binocular of size at most $\frac{4}{\epsilon'}\log |A_2|-1$ in $G_{A_2}$.

Let $G_A$ be the auxiliary multi-graph with vertices being $A$, every degree-1 vertex in $B$ creates a self-loop, every degree-2 vertex in $B$ creates an edge, and every vertex $v$ with degree dropping to 2 by performing some consistent tail changes of this vertex in $\mathcal{T}_c$ creates an edge between the two neighbors $u_1,u_2$ of $v$, where $(u_1,v),(u_2,v)$ are not associated with any tail change in $\mathcal{T}_c$. (Notice that contrary to the auxiliary multi-graph considered in \cite{ksetpacking2013}, here some vertices in $B$ might simultaneously create an edge in $G_A$ and involve in tail changes.) We have the following claim for sufficiently large $n$
($n>(\frac{k}{\epsilon})^{O(\epsilon)}$).

\begin{Claim}
\label{claimimprovement}
If there is a binocular of size $p\leq \frac{4}{\epsilon'}\log |A_2|-1$ in $G_{A_2}$, there exists a canonical improvement with tail changes in $G_A$ of size at most $\frac{12}{\epsilon'}\log n$.
\end{Claim}

Implied by the claim, we know that there exists a canonical improvement with tail changes in $G_A$ of size $\frac{12}{\epsilon'}\log n<\frac{4}{\epsilon}\log n$, which can be found by Algorithm LI. Therefore
\begin{equation}
\label{relationab2}
(1+\epsilon')(a-a_t)\geq (1+\epsilon')|A_2| \geq b_n^1+b_2^1+b^2_2+b_3^2.
\end{equation}

Combining (\ref{relationab1}) and (\ref{relationab2}), we have
\begin{eqnarray}
(k+1+\epsilon')(a-a_t) &\geq& 3b_n^1+\epsilon' b+2b_2^1+3b^2_2+3b_3^2+3b_3^3+2b_3^{2'}+b_3^1  \nonumber \\
&=& 3(b-b_t-\epsilon' b)-b_2^1-b_3^{2'}-2b_3^1+\epsilon' b.
\end{eqnarray}

Hence, $(3-2\epsilon')b \leq (k+1+\epsilon')a-(k-2+\epsilon')a_t+b_2^1+b_3^{2'}+2b_3^1$.

Since every vertex in $A_t$ can have at most $k-2$ edges to $B'$, we have $b_2^1+b_3^{2'}+2b_3^1 \leq (k-2)a_t$. Therefore, $(3-2\epsilon')b \leq (k+1+\epsilon')a$. As $\epsilon' =\frac{2k+5}{3}\epsilon$, we have $b\leq (\frac{k+1}{3}+\epsilon)a$. $\square$

\end{proof}

The proof of Claim \ref{claimimprovement} helps understand why we consider three more types of local improvements in addition to binoculars and helps explain the algorithm design in the next section.

\begin{comment}
\begin{proof}[Claim \ref{claimimprovement}]
Consider any binocular $I$ in $G_{A_2}$. If there is no edge in $I$ which is from $B_2^1$, we have a corresponding improvement $I'$ in $G_A$ by performing tail changes for any edge from $B_3^2$ in $I$.
Otherwise, we assume that there is one self-loop in $I$ from $v\in B_2^1$. By definition, one neighbor $u_1$ of $v$ lies in $A_t$ and the other neighbor $u_2$ in $A'$. Suppose $u_1$ belongs to a tail change in $\mathcal{T}_c$ which is associated with $w\in B_3^2$. If $w\in I$, we associate $w$ with tail changes in $G_A$. In $G_A$, we remove the self-loop on $u_2$ and add edge $(u_1,u_2)$. In this way, we have a path together with the other cycle in $I$ which form an improvement in $G_A$, assuming the other cycle in $I$ is not a self-loop from $B_2^1$.
If the other cycle in $I$ is also a self-loop from $v'\in B_2^1$, let $u_1'$ be a neighbor of $v'$ in $A_t$ and $u_2'$ be the other neighbor of $v'$ in $A'$. If $u_1'$ belongs to the tail change associated with $w'\in B_3^2$ and $w'\in I$, the path between $u_2,u_2'$ in $I$ together with the edges $(u_1,u_2),(u_1',u_2')$ form an improvement. If $u_1=u_1'$, we have an improvement in $G_A$ as a cycle. Other cases can be analyzed similarly. $\square$
\end{proof}
\end{comment}

\begin{proof}[Claim \ref{claimimprovement}]
There are two differences between $G_{A_2}$ and $G_A$. First, there is no tail changes involved in any improvement from $G_{A_2}$. While in $G_A$, if we want to include edges which come from vertices of degree at least 3 in $B$ into an improvement, we also need to perform tail changes. Second, any vertex in $B_2^1$ creates a self-loop in $G_{A_2}$, while in $G_A$ it creates an edge.

Consider any binocular $I$ in $G_{A_2}$ which forms a pattern in the first row of Figure 1. If there is no edge in $I$ which is from $B_2^1$, we have a corresponding improvement $I'$ in $G_A$ by performing tail changes for any edge from $B_3^2$ in $I$. As we select a consistent collection of tail changes $\mathcal{T}_c$, $I'$ is a valid improvement.

Otherwise, we first assume that there is one self-loop in $I$ from $v\in B_2^1$. By definition, one neighbor $u_1$ of $v$ lies in $A_t$ and the other $u_2$ in $A'$. Suppose $u_1$ belongs to a tail change in $\mathcal{T}_c$ which is associated with $w\in B_3^2$. By the matching process and Corollary \ref{corofmatchprocess}, there exists a path from $u_1$ to a vertex $b_1\in B_L^1$, where the edges in this path might come from degree-2 vertices, or from higher degree vertices with consistent tail changes. Let $a_1$ be the neighbor of $b_1$.

If $w\notin I$, we have a canonical improvement in $G_A$ by replacing the self-loop from $v$ by the path from $u_2,u_1$ to $a_1$ and a self-loop on $a_1$ from $b_1$. The size of the improvement increases by at most $1+\log_{k-1}\frac{1}{\epsilon}$.
If $w\in I$, we associate $w$ with tail changes in $G_A$. We replace the self-loop on $u_2$ from $v$ to the edge $(u_1,u_2)$. In this way, we have a path together with the other cycle in $I$ which form an improvement in $G_A$, assuming the other cycle in $I$ is not a self-loop from $B_2^1$.

Finally, if the other cycle in $I$ is also a self-loop from $v'\in B_2^1$, let $u_1'$ be a neighbor of $v'$ in $A_t$ and $u_2'$ be another neighbor of $v'$ in $A'$. If $u_1'$ belongs to the tail change associated with $w'\in B_3^2$ and $w'\in I$, the path between the two self-loops in $I$ together with the edges $(u_1,u_2),(u_1',u_2')$ form an improvement (as there are tail changes involved). Here $w'$ and $w$ could be the same vertex. Otherwise, we can also replace the self-loop of $v'$ by a path with one endpoint attaching a new self-loop. In this way, there is an improvement in $G_A$ with size increasing by at most $2(1+\log_{k-1}\frac{1}{\epsilon})$. Notice that if the two new self-loops are the same, two new paths and the original path between $v,v'$ form a cycle.

Notice that $1+\log_{k-1}\frac{1}{\epsilon}<\frac{4}{\epsilon'}\log n$ for $n>(\frac{k}{\epsilon})^{O(\epsilon)}$. Therefore, for any binocular of size $p\leq \frac{4}{\epsilon'}\log |A_2|-1$ in $G_{A_2}$, there exists a corresponding canonical improvement in $G_A$ of size at most $\frac{12}{\epsilon'}\log n$.   $\square$

\end{proof}

\section{The Algorithm and Main Results}

In this section, we give an efficient implementation of Algorithm LI using the color coding technique \cite{colorcoding} and dynamic programming. Let $U$ be the universe of elements and $K$ be a collection of $kt$ colors, where $t=\frac{4}{\epsilon}\log n \cdot \frac{2(k-1)}{\epsilon}\cdot (k-2)\leq \frac{4}{\epsilon}\log n \cdot \frac{2k^2}{\epsilon}$. We assign every element in $U$ one color from $K$ uniformly at random. If two $k$-sets contain $2k$ distinct colors, they are recognized as disjoint. Applying color coding is crucial to obtain a polynomial-time algorithm for finding a logarithmic-sized local improvement.

%We then summarize our new polynomial-time local search algorithm with approximation ratio $\frac{k+1}{3}+\epsilon$.

\subsection{Efficiently finding canonical improvements with tail changes}

In this section, we show how to efficiently find canonical improvements with tail changes using the color coding technique. Let $C(S)$ be the set of distinct colors contained in sets in $S$. We say a collection of sets is colorful if every set contains $k$ distinct colors and every two sets contain different colors.

{\bf Tail changes.} We say a tail change $T_e(U,V)$ of a vertex $v$ is {\it colorful} if $V$ is colorful, and the colors in $C(V)$ are distinct from $C(v)$. A surplus edge can be associated with many tail changes. Let $\mathcal{T}_v(e)$ be all colorful tail changes of size at most $\frac{2(k-1)}{\epsilon}$ which are associated with an edge $e$ of $v$. We enumerate all subsets of $\mathcal{S}\setminus \mathcal{A}$ of size at most $\frac{2(k-1)}{\epsilon}$ and check if they are colorful, and if they are tail changes associated with $e$. The time to complete the search for all vertices is at most $n^{O(k/\epsilon)}$.

The next step is to find all colorful groups of tail changes associated with $v$ such that after performing one group of tail changes, the degree of $v$ drops to 2. Notice that the tail changes in a colorful group are consistent. For every two edges $e_i,e_j$ of $v$, we can compute a collection of colorful groups of tail changes which associate with all edges of $v$ except $e_i, e_j$ by comparing all possible combinations of tail changes from $E(v)\setminus \{e_i,e_j\}$. There are at most $(n^{O(k/\epsilon)})^{k-2}$ combinations. For every group of colorful tail changes which together replace $V$ with $U$, we explicitly keep the information of which vertices are in $U,V$ and the colors of $V$. It takes at most $n^{O(k^2/\epsilon)}$ space.
To summarize, the time of finding colorful groups of tail changes for every vertex of degree at least 3 is $n^{O(k^2/\epsilon)}$.

{\bf Canonical improvements with tail changes.} After finding all colorful tail changes for every vertex of degree at least 3, we construct the auxiliary multi-graph $G_{A}$. For vertices $a_1,a_2$ in $G_{A}$, we put an edge $e(a_1,a_2)$ between $a_1$ and $a_2$ if first, there is a set $b\in \mathcal{C}=\mathcal{S}\setminus \mathcal{A}$ intersecting with only $a_1,a_2$, or secondly, there is a set $b\in \mathcal{C}$ of degree $d_b\geq 3$ intersecting with $a_1,a_2$, and for other edges of $b$, there exists at least one group of $d_b-2$ colorful tail changes. In the first case, we assign the colors of $b$ to $e(a_1,a_2)$. In the second case, we add as many as $n^{O(k^2/\epsilon)}$ edges between $a_1$ and $a_2$, and assign to each edge the colors of $b$ together with the colors of the corresponding group of $d_b-2$ tail changes. The number of edges between two vertices in $G_{A}$ is at most $n\cdot n^{O(k^2/\epsilon)}$. The number of colors assigned to each edge is at most $\frac{2k^3}{\epsilon}$ (Notice that the number of colors on an edge is at most $k(1+\frac{(k-1)^2/\epsilon-1}{k-2})(k-2)$. This is at most $\frac{2k^3}{\epsilon}$ for $\epsilon<k+5$, which is usually the case.) Moreover, we add a self-loop for a vertex $a$ in $G_A$ if there exists a set $b\in\mathcal{C}$ such that $b$ intersects only with set $a$ and assign the colors of $b$ to this self-loop.

We use the dynamic programming algorithm to find all colorful paths and cycles of length $p=\frac{4}{\epsilon}\log n$ in $G_{A}$. A path/cycle is colorful if all the edges contain distinct colors. If we did not consider improvements containing at most one cycle, we could use a similar algorithm as in \cite{ksetpacking2013}. In our case, when extending a path by an edge $e$ by dynamic programming, we would like to keep the information of the vertices replaced by the tail changes of $e$ in this path. This would take quasi-polynomial time when backtracking the computation table. By Claim \ref{claimimprovement}, it is sufficient to check for every path with endpoints $u,v$, if there is an edge of this path containing a tail change $T_e(U,V)$, such that $u\in U$ or $v\in U$. We sketch the algorithm as follows. For a given set of colors $C$, let $\mathcal{P}(u, v, j, C, q_u, q_v)$ be an indicator function of whether there exists a path of length $j$ from vertex $u$ to $v$ with the union of the colors of these edges equal $C$. $q_u$($q_v$) are indicator variables of whether there is a tail change $T_e(U,V)$ of some edge in the path such that $u\in U$($v\in U$). The computation table can be initialized as $\mathcal{P}(u, u, 0, \emptyset, 0, 0)=1$ and $\mathcal{P}(u, v, 0, \emptyset, 0, 0)=1$, for every $u,v\in A$. In general, for a fixed set of colors $C$ and integer $j\geq 1$, $\mathcal{P}(u, v, j, C, q_u, q_v)=1$ if there exists a neighbor $w$ of $v$, such that $\mathcal{P}(u, w, j-1, C', q_u, q_w)=1$, $C'\cup C((w,v))= C$ and $C'\cap C((w,v))=\emptyset$. If $C((w,v))>k$ (i.e., there are tail changes), we check every edge between $w$ and $v$ which satisfies the previous conditions. If there exists an edge associated with a tail change $T_e(U,V)$ such that $u,v\in U$, we mark $q_u=1,q_v=1$. Otherwise, if there exists an edge associated with a tail change $T_e(U,V)$ such that $u\in U$, we mark $q_u=1$. To find colorful cycles, we query the result of $\mathcal{P}(u,u,j,C, q_u, q_u)$ for $j\geq 1$. Recall that we use $kt$ many colors, where $t\leq p\cdot \frac{2k^2}{\epsilon}$. The running time of finding all colorful paths and cycles is $O(n^3kp 2^{kt})$, which is $n^{O(k^3/\epsilon^2)}$.

The final step is to find canonical improvements with tail changes by combining colorful paths and cycles to form one of the six types defined in Section 3.1 by enumerating all possibilities. The running time of this step is $n^{O(k^3/\epsilon^2)}$.
In conclusion, the total running time of finding colorful tail changes, colorful paths/cycles, and canonical improvements with tail changes is $n^{O(k^3/\epsilon^2)}$. We call this color coding based algorithm {\bf Algorithm CITC} (canonical improvement with tail changes). The running time analysis of Algorithm CITC is given in the appendix.

\subsection{Main results}

In this section, we present our main results. We first present a randomized local improvement algorithm.
The probability that Algorithm CITC succeeds in finding a canonical improvement with tail changes if one exists can be calculated as follows. The number of sets involved in a canonical improvement with tail changes is at most $\frac{2k^2}{\epsilon}\cdot \frac{4\log n}{\epsilon}$. The probability that an improvement with $i$ sets having all $ki$ elements of distinct color is
\begin{equation}
\label{colorsuccesspr}
\frac{{kt\choose ki}(ki)!}{(kt)^{ki}} = \frac{(kt)!}{(kt-ki)!(kt)^{ki}}\geq \frac{(kt)!}{(kt)^{kt}} >e^{-kt} \geq n^{-8k^3/\epsilon^2}.
\end{equation}

Let $N=n^{8k^3/\epsilon^2}\ln n$. We run Algorithm CITC $2N$ times and each time with a fresh random coloring. From (\ref{colorsuccesspr}), we know that the probability that at least one call of CITC succeeds in finding an improvement is at least 
\begin{equation}
1-(1-n^{-8k^3/\epsilon^2})^{2N} \geq 1-exp(n^{-8k^3/\epsilon^2}\cdot 2n^{-8k^3/\epsilon^2}\ln n) = 1-n^{-2}. \nonumber
\end{equation}

Since there are at most $n$ local improvements for the problem, the probability that all attempts succeed is at least $(1-n^{-2})^n \geq 1-n^{-1}\longrightarrow 1$ as $n\longrightarrow \infty$.
Hence this randomized algorithm has an approximation ratio $\frac{k+1}{3}+\epsilon$ with high probability. We call this algorithm {\bf Algorithm RLI} (R for randomized). The running time of the algorithm is $2N\cdot n^{O(k^3/\epsilon^2)}$, which is $n^{O(k^3/\epsilon^2)}$.

\begin{comment}
We summarize Algorithm RLI3 as follows.
\renewcommand{\thealgorithm}{RLI3}
\begin{algorithm}
\caption{Randomized implementation of canonical improvements with tail changes for the $k$-Set Packing problem}
\label{localsearch}
\begin{algorithmic}[1]
\State {\bf Input:} a collection of $k$-sets $\mathcal{S}$ on universe $U$, parameter $\epsilon>0$.
\State Pick a packing $\mathcal{A}$ greedily.
\Repeat
\State Perform all local improvements of size $O(\frac{k}{\epsilon})$.
\State {\bf Loop} $2n^{8k^3/\epsilon^2}\ln n$ times:
\State - Assign $kt$ colors to $U$ uniformly at random, $t=\frac{2(k-1)(k-2)}{\epsilon}\cdot \frac{4\log n}{\epsilon}$.
\State - Find all colorful groups of tail changes for sets in $\mathcal{S}\setminus\mathcal{A}$ of degree at least 3.
\State - Construct the auxiliary multi-graph and find canonical improvements by dynamic programming.
\State - Update $\mathcal{A}$ if an improvement is found.

\Until{there is no local improvement.}

\State {\bf Return} $\mathcal{A}$.
\end{algorithmic}
\end{algorithm}

\end{comment}

We can obtain a deterministic implementation of Algorithm RLI, which always succeeds in finding a canonical improvement with tail changes if one exists. We call this deterministic algorithm {\bf Algorithm DLI} (D for deterministic). The general approach is given by Alon et al. \cite{derandom}. The idea is to find a collection of colorings $\mathcal{K}$, such that for every improvement there exists a coloring $K\in \mathcal{K}$ that assigns distinct colors to the sets involved in this improvement. Then Algorithm CITC can be implemented on every coloring until an improvement is found. The collection of colorings satisfying this requirement can be constructed using perfect hash functions from $U$ to $K$. A perfect function for a set $S\subseteq U$ is a mapping which is one-to-one on $S$. A $p$-perfect family of hash functions contains one perfect hash function for each set $S\subseteq U$ of size at most $p$. Alon and Naor \cite{derandom} show how to construct a perfect hash function from $[m]$ to $[p]$ in time $O(p\log m)$ explicitly. This function can be described in $O(p+\log p\log\log m)$ bits and can be evaluated in $O(\log m/\log p)$ time.
In our case, we use a $kt$-perfect family of perfect hash functions from $U$ to every $K\in \mathcal{K}$. $\mathcal{S}$ covers at most $nk$ elements. The number of the perfect hash functions in this family is at most $2^{kt+\log kt\log\log nk}$, which is $n^{O(k^3/\epsilon^2)}$. Hence, we need $n^{O(k^3/\epsilon^2)}$ runs of dynamic programming to find an improvement.

\begin{theorem}
For any $\epsilon>0$, Algorithm DLI achieves an approximation ratio $\frac{k+1}{3}+\epsilon$ of the $k$-Set Packing problem in time $n^{O(k^3/\epsilon^2)}$.
\end{theorem}

\section{Lower bound}

In this section, we construct an instance with locality gap $\frac{k+1}{3}$ such that there is no local improvement of size up to $O(n^{1/5})$. Our construction is randomized and extends from the lower bound construction in \cite{ksetpacking2013}. This lower bound matches the performance guarantee of Algorithm DLI.

\begin{theorem}
\label{lowerbound}
For any $t\leq \left(\frac{3e^3n}{k}\right)^{1/5}$, there exist two disjoint collections of $k$-sets $\mathcal{A}$ and $\mathcal{B}$ with $|\mathcal{A}|=3n$ and $|\mathcal{B}|=(k+1)n$, such that any collection of $t$ sets in $\mathcal{B}$ intersect with at least $t$ sets in $\mathcal{A}$.
\end{theorem}

\begin{proof}%[Proof of Theorem \ref{lowerbound}]
Consider a universe $U_A$ of $3kn$ elements. Let $\mathcal{A}$ be a collection of $a=3n$ disjoint $k$-sets on $U_A$. We index the sets in $\mathcal{A}$ by 1 to $3n$. Let $\mathcal{B}$ be a collection of $(k+1)n$ disjoint $k$-sets, such that every set induced on $U_A$ is a 2-set or a 3-set. There are $b_2=3n$ 2-sets covering $m_2=6n$ elements and $b_3=(k-2)n$ 3-sets covering $m_3=3kn-6n$ elements in $\mathcal{B}$. We index the 2-sets in $\mathcal{B}$ by 1 to $3n$. The $i$-th 2-set intersects with the $(i-1)$-th and the $i$-th set in $\mathcal{A}$ (the 0-th set is the $n$-th set in $\mathcal{A}$). The 3-sets are constructed by partitioning the elements not covered by 2-sets in $U$ into groups of three uniformly at random.

Consider an arbitrary collection $\mathcal{A}_t$ of $t$ sets in $\mathcal{A}$. We compute the probability that there are $t$ sets $\mathcal{B}_t$ in $\mathcal{B}$ which are contained entirely in $\mathcal{A}_t$. We call such an event {\it unstable}. Assume there are $t_2$ 2-sets $\mathcal{B}_t^2$ and $t_3$ 3-sets $\mathcal{B}_t^3$ in $\mathcal{B}_t$. Suppose the sets in $\mathcal{A}_t$ belong to $r$ disjoint index intervals, $i_1,....,i_r\geq 1$ for $1\leq r\leq t$. Here we use modular $3n$ to compute the connectivity. Then $t_2=i_1-1+i_2-1+\cdots +i_r-1=i_1+\cdots+i_r-r=t-r$ and $\mathcal{B}_t^2$ cover $2(t-r)$ elements. $t_3=r$ and $\mathcal{B}_t^3$ cover $3r$ elements from a set of $m_t=kt-2(t-r)=(k-2)t+2r$ elements.

Let $\tau(m)$ be the number of ways to partition $m$ elements into $m/3$ disjoint sets. We have
\begin{equation}
\label{tau}
\tau(m)=\frac{m!}{(3!)^{m/3}(m/3)!}.
\end{equation}

Let $\Pr(t,r)$ be the probability that all $t_3$ 3-sets are contained in $\mathcal{A}_t$. We have
\begin{equation}
\label{prunstable}
\Pr(t,r) = {m_t\choose 3r}\cdot \frac{\tau(3r)\tau(m_3-3r)}{\tau(m_3)}=\frac{{m_t\choose 3r}\cdot{m_3/3\choose r}}{{m_3\choose 3r}}.
\end{equation}

Let $U_{t,r}$ be the number of all unstable events summing over the distribution of $\mathcal{B}$. There are ${t-1\choose r-1}$ positive integer solutions of the function $t=i_1+i_2+\cdots +i_r$. There are ${a-t\choose r}$ intervals of length $i_1,...,i_r$ in index range 1 to $a$. Hence, the expected number of unstable events can be estimated as follows,

\begin{eqnarray}
\label{expunstable1}
\mathbb{E}[U_{t,r}]&=&{t-1\choose r-1} \cdot {a-t\choose r} \cdot \Pr(t,r)
\approx {t\choose r} \cdot {a-t\choose r} \cdot \frac{{(k-2)t+2r\choose 3r}\cdot{(k-2)a/3\choose r}}{{(k-2)a\choose 3r}} \nonumber \\
&\leq& \left(\frac{\frac{t}{r}\cdot\frac{a-t}{r}\cdot(\frac{(k-2)t+2r}{3r})^3\cdot \frac{(k-2)a}{3r}}{(\frac{e(k-2)a}{3r})^3} \right)^r = \left(\frac{t(a-t)((k-2)t+2r)^3}{3e^3r^3(k-2)^2a^2} \right)^r
\end{eqnarray}

The inequality in (\ref{expunstable1}) follows from an upper and lower bounds of the combinatorial number as

\begin{equation}
\label{nchoosekbound}
\left(\frac{n}{k}\right)^k \leq {n\choose k} \leq \left(\frac{en}{k}\right)^k.
\end{equation}

$\mathbb{E}[U_{t,r}]$ can be further bounded from (\ref{expunstable1}) as,
\begin{eqnarray}
\label{expunstable2}
\mathbb{E}[U_{t,r}] &=& \left(\frac{t(a-t)}{3e^3(k-2)^2a^2}\cdot\left(\frac{(k-2)t}{r}+2\right)^3 \right)^r \nonumber \\
&\leq& \left(\frac{t}{3e^3(k-2^2a^2)}\cdot\frac{k^3t^3}{r^3}\right)^r \leq \left(\frac{kt^4}{3e^3ar^3}\right)^r.
\end{eqnarray}

Since $t\leq \left(\frac{3e^3n}{k}\right)^{1/5}=\left(\frac{e^3a}{k}\right)^{1/5}=t_0$ by assumption, we have $\frac{kt^4}{3e^3ar^3}\leq \frac{kt^4}{3e^3a}\leq \frac{1}{3}(\frac{k}{e^3a})^{1/5}<\frac{1}{2}$ as $a\gg 1$. Hence by summing up $r$ from 1 to $t$ and $t$ from 1 to $t_0=\left(\frac{e^3a}{k}\right)^{1/5}$ in (\ref{expunstable2}), we have,
\begin{eqnarray}
\label{expunstable3}
\sum_{t=1}^{t_0}\sum_{r=1}^t \mathbb{E}[U_{t,r}] &<& \sum_{t=1}^{t_0} \sum_{r=1}^t \left(\frac{kt^4}{3e^3a}\right)^r < \sum_{t=1}^{t_0} \frac{2kt^4}{3e^3a} < \frac{2k}{3e^3a}\cdot t_0^5 <1.
\end{eqnarray}

Therefore, there exists some collection $\mathcal{B}$ given $\mathcal{A}$ which does not contain any unstable collection of size at most $t$.  $\square$
\end{proof}

\section{Conclusion}

In this paper, we propose a new polynomial-time local search algorithm for the $k$-Set Packing problem and show that the performance ratio is $\frac{k+1}{3}+\epsilon$ for any $\epsilon>0$. While the approximation guarantee is the same as for Cygan's algorithm \cite{bestksetpacking}, our algorithm has a better running time, which is singly exponential to $\frac{1}{\epsilon^2}$. We also give a matching lower bound, which shows that any algorithm using local improvements of size at most $O(n^{1/5})$ cannot have a better performance guarantee. This indicates that this is possibly the best result that can be achieved by a local improvement algorithm for the $k$-Set Packing problem. On the other hand, algorithms based on LP/SDP for the $k$-Set Packing problem are far from being well understood. It is interesting to explore possibilities using that approach. A more general open question is to further close the gap between the upper bound $\frac{k+1}{3}+\epsilon$ and the lower bound $\Omega(\frac{k}{\log k})$ of the $k$-Set Packing problem \cite{ksetpackinglowerbound}.

%\newpage
\bibliographystyle{plain}
\bibliography{ksetpackingisco}

\newpage
\section*{Appendix}

%\appendix

We give detailed running time analysis of Algorithm CITC in this appendix. \\

{\bf Running time of finding colorful tail changes.}
The time to complete the search of all tail changes for every edges by enumerating subsets in $\mathcal{S}\setminus \mathcal{A}$ of size at most $\frac{2(k-1)}{\epsilon}$ is at most
\begin{equation}
\label{timetailchange}
n^2\cdot \sum_{i=1}^{2(k-1)/\epsilon} {n\choose i}\cdot (i+1)k \leq \frac{2k^2}{\epsilon} n^{\frac{2(k-1)}{\epsilon}+2}=n^{O(\frac{k}{\epsilon})}.
\end{equation}

The time to find all consistent tail changes for every vertex of degree at least 3 is in the order of
\begin{equation}
\label{timeconsistenttc}
n{k\choose 2}\cdot (n^{O(\frac{k}{\epsilon})})^{k-2} \cdot O(\frac{k^2}{\epsilon}) = n^{O(\frac{k^2}{\epsilon})}.
\end{equation}

Combining (\ref{timetailchange}) and (\ref{timeconsistenttc}), the time of finding colorful groups of tail changes for every vertex of degree at least 3 is $n^{O(\frac{k^2}{\epsilon})}$. \\

In the following, the length of a path/cycle is at most $p=\frac{4}{\epsilon}\log n$. The number of colors is at most $kt$, where $t=p\cdot \frac{2(k-1)(k-2)}{\epsilon}$. This part is similar as \cite{ksetpacking2013}. \\

{\bf Running time of finding colorful paths.} There are at most $2^{kt}$ sets of colors. For every path of length $j$, there are at most $n^2 2^{kt}$ entries in $\mathcal{P}(u,v,j,C)$. It takes time $O(nk)$ to check if $C'\cup C((w,v))=C$ and $C'\cap C((w,v))=\emptyset$. Hence, the total time to fill in the computation table is in the order of
\begin{equation}
\label{timepath}
\sum_{j=1}^p n^2 2^{kt} \cdot nk = n^3kp 2^{kt}.
\end{equation}

{\bf Running time of finding canonical improvements.}
\begin{enumerate}
    \item The first type of canonical improvement. For every vertex $s\in A$, and for all disjoint collections $C_1,C_2$ of $kc_1,kc_2$ colors respectively, such that $c_1+c_2\leq t$, check if there exists integers $l_1,l_2$ with $l_1+l_2\leq p$, such that $\mathcal{P}(s,s,l_1,C_1)=1$ and $\mathcal{P}(s,s,l_2,C_2)=1$. This takes time in the order of
        \begin{equation}
        \label{timeci1}
        n\sum_{c_2=1}^{t}\sum_{c_1=1}^{t-c_2} 2^{kc_1}2^{kc_2}p^2\cdot tk \leq np^2t^3 k 2^{kt}.
        \end{equation}

    \item The second type of canonical improvement. For all pairs of vertices $s,t\in A$, and for all disjoint collections $C_1,C_2,C_3$ of $kc_1,kc_2,kc_3$ colors respectively, such that $c_1+c_2+c_3\leq t$, check if there exists integers $l_1,l_2,l_3$ with $l_1+l_2+l_3\leq p$, such that $\mathcal{P}(s,s,l_1,C_1)=1$, $\mathcal{P}(t,t,l_2,C_2)=1$ and $\mathcal{P}(s,t,l_3,C_3)=1$. This takes time in the order of
        \begin{equation}
        \label{timeci2}
        n^2\sum_{c_3=1}^t\sum_{c_2=1}^{t-c_3}\sum_{c_1=1}^{t-c_2-c_3} 2^{kc_1}2^{kc_2}2^{kc_3}p^3 \cdot tk \leq n^2p^3 t^4 k 2^{kt}.
        \end{equation}

    \item The third type of canonical improvement. For all pairs of vertices $s,t\in A$, and for all disjoint collections $C_1,C_2,C_3$ of $kc_1,kc_2,kc_3$ colors respectively, such that $c_1+c_2+c_3\leq t$, check if there exists integers $l_1,l_2,l_3$ with $l_1+l_2+l_3\leq p$, such that $\mathcal{P}(s,t,l_1,C_1)=1$, $\mathcal{P}(s,t,l_2,C_2)=1$ and $\mathcal{P}(s,t,l_3,C_3)=1$. This takes time in the order of
        \begin{equation}
        \label{timeci3}
        n^2\sum_{c_3=1}^t\sum_{c_2=1}^{t-c_3}\sum_{c_1=1}^{t-c_2-c_3} 2^{kc_1}2^{kc_2}2^{kc_3}p^3 \cdot tk \leq n^2p^3 t^4 k 2^{kt}.
        \end{equation}

    \item The fourth to sixth type of canonical improvement. The running time is dominated by the fifth type, namely when a path from $S$ to $T$ and a cycle of $S$ form an improvement. For every vertex $s,t\in A$, and for all disjoint collections $C_1,C_2$ of $kc_1,kc_2$ colors respectively, such that $c_1+c_2\leq t$, check if there exists integers $l_1,l_2$ with $l_1+l_2\leq p$, such that $\mathcal{P}(s,s,l_1,C_1)=1$ and $\mathcal{P}(s,t,l_2,C_2)=1$. Moreover, backtrack in the computation table, check if the sets indeed form an improvement. This takes time in the order of
        \begin{equation}
        \label{timeci4}
        n^2\sum_{c_2=1}^{t}\sum_{c_1=1}^{t-c_2} 2^{kc_1}2^{kc_2}p^2\cdot tk \leq n^2p^2t^3 k 2^{kt}.
        \end{equation}

\end{enumerate}

\end{document}